\documentclass{article}
\usepackage{ijcai19}

\usepackage{times}
\usepackage{textcomp}
\usepackage{url}
\usepackage[utf8]{inputenc}
\usepackage[numbers, compress, sort, square]{natbib}
\usepackage{subcaption}
\usepackage{tikz}
\usetikzlibrary{shapes,shapes.geometric,arrows,fit,calc,positioning,automata}
\usepackage{amssymb,amsmath,amsthm}
\usepackage{tikz-cd}

\newtheorem{theorem}{Theorem}
\newtheorem{lemma}{Lemma}
\newtheorem{definition}{Definition}

\newcommand*{\field}[1]{\mathbb{#1}}%
\newcommand{\N}{\field{N}}
\newcommand{\R}{\field{R}}
\newcommand{\DD}{\mathcal{D}}
\newcommand{\err}{\mathrm{err}}
\newcommand{\acc}{\mathrm{acc}}

\newcommand{\x}{{m_1}} 
\newcommand{\y}{{m_2}} 
\renewcommand{\a}{{a}} 
\renewcommand{\b}{{b}} 
\newcommand{\n}{{n}} 
\newcommand{\p}{{p}} 
\newcommand{\ms}{{\mu}} 
\newcommand{\util}{{u}} 
\newcommand{\cons}{{CW}} 

\newcommand{\xna}{(\x+\n)^\b}
\newcommand{\yna}{(\y+\n)^\b}
\newcommand{\xnaya}{\xna + \y^\b}

\newcommand{\xayna}{\x^\b+ \yna}
\newcommand{\xaya}{\x^\b + \y^\b}
\newcommand{\xa}{\x^\b}
\newcommand{\ya}{\y^\b}
\newcommand{\A}{{A}}
\newcommand{\C}{{C}}
\newcommand{\D}{{D}}
\newcommand{\HH}{\mathcal{H}}

\title{Equilibrium Characterization for Data Acquisition Games}

\author{
Jinshuo Dong\and Hadi Elzayn\and Shahin Jabbari\and
Michael Kearns\And Zachary Schutzman
\affiliations
University of Pennsylvania\\
\emails
\{jinshuo, hads\}@sas.upenn.edu,
\{jabbari, mkearns, ianzach\}@cis.upenn.edu
}

\begin{document}
\maketitle
\begin{abstract}
We study a game between two firms in which each provide a 
service based on machine learning.  The firms are presented 
with the opportunity to purchase a new corpus of data, which 
will allow them to potentially improve the quality of their products. 
The firms can decide whether or not they want to buy the data, as 
well as which learning model to build with that data. 
We demonstrate a reduction from this potentially complicated action 
space  to a one-shot, two-action game in which each firm only decides 
whether or not to buy the data. The  
game admits several regimes which depend on the relative strength of 
the two firms at the outset and the price at which the data is 
being offered. We analyze the game's Nash equilibria in all parameter 
regimes and demonstrate 
that, in expectation, the outcome of the game is that the initially stronger 
firm's market position weakens whereas the initially 
weaker firm's market position becomes stronger.  
Finally, we consider the perspective of the users of the service and 
demonstrate that the expected outcome at equilibrium is not the one 
which maximizes the welfare of the consumers.
\end{abstract}
\section{Introduction}
\label{sec:intro}
Recent years have seen explosive growth in the domain of digital data-driven services.  
Search engines, restaurant recommendations, and social media are
among the many products we use day-to-day which sit atop modern data analysis and 
machine learning (ML). In such markets, firms live and die by the quality of their models; 
thus success in the `race for data', whether acquired directly from customers or indirectly 
via acquisition of rival firms or purchasing data corpuses, is crucial. In this work, we study 
two questions: whether such markets tend towards monopoly, and how competition affects 
consumer welfare. Importantly, we consider these questions in light of the modeling choices 
that firms must make. 

In our model, two firms compete for market share (utility) by providing identical services that 
each rely on an ML model. The firms' error rates depend on their  
choices of algorithms, models and the volume of available training data. 
Each firm's market share is proportional to the error of its model relative to 
the model built by its competitor. This is motivated 
by the observation that the services built using ML are highly 
accurate, so users are more conscientious of the mistakes the service 
makes, rather than the successes. A competition exponent measures 
relative ferocity of competition and maps to a plausible Markov model 
of consumer choice. See Section~\ref{sec:dynamics} for more details. 

The firms initially possess (possibly differing) quantities of data,  
and are given the opportunity to buy additional data at a fixed price 
to improve their models. Since data is costly and \emph{relative} 
(rather than absolute)  model quality 
determines market share, each firm's best course of action may 
depend on the actions of its rival.
Hence, each firm acts strategically and faces two decisions: whether 
to buy the additional data, and what type of model to build in order to 
produce the best product given the data it ends up with. 

The decision of what model to build seems to complicate the firms' 
action space greatly; there is a very large set of model classes to select 
from, and different classes have different efficiencies. For example, when 
restricting attention to neural networks, the choices of depth and number of 
nodes per layer produce different hypothesis classes with different optimal 
models. Thus, in principle, the decisions of what model class to select and whether to purchase 
additional data must be made jointly. However, learning theory allows 
us to greatly reduce this large action space. In Section \ref{sec:ml}, we show 
that the game in which firms jointly choose a model and whether to attempt to 
buy the additional data reduces to a strategically equivalent game in which firms 
first choose whether to buy the data and then choose optimal models.

In Section \ref{sec:eq-analysis}, we characterize the Nash equilibria 
of our game for different parameter regimes. For no combination of 
parameters does exactly one firm wish to buy the data; unsurprisingly, 
for very high prices, neither firm buys data, and for very low prices, both 
firms do. In the middling regime, the competitive aspect of the game imposes 
a `prisoners\textquotesingle \ dilemma'-like flavor: both firms would 
prefer neither firm buy the data, but each do so in order to prevent the other 
from strengthening its position.
Moreover, the unique mixed strategy Nash equilibrium in this regime 
involves firms \emph{increasing} their probability of buying data as price 
\emph{increases}. This counterintuitive result follows from the logic of 
equilibrium: firms playing mixed strategies must be indifferent to buying 
and not buying the data, and as the price rises, the probability that a competitor 
acquires the data must rise in order to make investing in data acquisition a 
palatable option.  

Finally, we study whether any of the dynamics of the game push the market 
towards a monopoly. Perhaps counter to a `rich-get-richer' feedback loop 
that might be expected in data races, we observe 
that in all equilibria, the data gap (and thus, market share gap) always 
narrows (in expectation). As measured by consumer welfare, this is 
actually \emph{undesirable}. Both the direction of the data gap as well 
as the welfare implication may be counterintuitive, particularly 
with respect to the well-known stylized fact that market concentration is 
bad for consumers. However, consumer data that improves a service can 
be viewed as exhibiting a form of network effects, in which case perfect competition can 
result in inefficiency and under-provisioning of a good~\cite{katz1985network}. 
In other words, a greater data gap would result in more consumers using 
a less error-prone service. As for the data race, anecdotal evidence, such as 
GM's acquisition of automated driving startup Cruise, despite Waymo's earlier market 
entry and research head-start, are suggestive (though not conclusive) that these 
predictions may be indicative of real-world dynamics~\cite{Primack16}.

We view our work as a first step towards modeling and analyzing competition for data in 
markets driven by ML. 
Under our simplifying assumptions, we derive concrete results with 
relevance both for policymakers analyzing algorithmic actors as well as engineering or business 
decision-makers considering the tasks of data acquisition and model selection. 
Our results are qualitatively robust to other natural 
modeling choices, such as allowing both firms to purchase the data, as well as treating 
the data seller as a market participant; however, more significant departures may lead to 
different conclusions. See Sections~\ref{sec:extensions}~and~\ref{sec:future}
for more details. 
\subsection{Related Work}
\label{sec:related}
The theory of ML from a single learner's perspective is well developed, but until recently,
little work had studied competition between learning algorithms. Notable exceptions include \cite{wu, bpt}.
We differ from both works by exploring the comparative statics and welfare consequences of a 
single decision (data acquisition). Concurrent work \cite{benporat2019}  studies a game in which learners 
strategically choose their model to compete for users, but users only care about the accuracy of predictions on 
their particular data. In contrast, users in our model choose based on the overall model error. 

Our work also intersects with several strains of economic literature, including industrial 
organization and network effects~\cite{katz1985network,david1987some,economides1996economics}. 
We differ from such models in two key ways. First, in contrast to 
assuming a static equilibrium~\cite{katz1985network} or fixing 
a dynamic but unchanging process at the outset~\cite{salfar}, 
our work can be viewed as an analysis of a shock to a given potentially asymmetric equilibrium 
in the form of the availability of new data. Second, the consumers in our model do not behave 
strategically (see e.g.~\cite{vohra,wu} for more discussion).

Finally, our work is related to spectrum auctions, competition with congestion externalities~\cite{vohra}, 
and the sale of information or patents~\cite{kamien1992optimal,kamien1986fees}. 
Our results primarily share qualitative similarities: the choice of one firm to buy data (spectrum) forces the other to 
do so to avoid losing market share, though it would not have been profitable absent the rival, and actual 
outcomes run counter to consumer preferences (see e.g.~\cite{vohra}). 

\section{Framework}
\label{sec:framework}

We formally motivate and model the ML problem of the firms and demonstrate 
how this reduces to a game in which the firms can either buy or not buy the new data. 

\subsection{Choosing a Model Class}
\label{sec:ml}
Consider a firm using ML to build a service e.g. 
a recommendation system.
The amount of data available to the firm is a crucial determinant to the effectiveness 
of the predictive service of the firm. 
Fixing the amount of data, the firm faces a fundamental tradeoff; it can use a more complex model  
that can fit the data better, but learning using a complicated model requires more training data to 
avoid over- or underfitting.

We can formally represent this tradeoff as follows. Let $\HH$ denote the hypothesis class from which the firm
is selecting its model and assume the data is generated from a distribution $\DD$. 
Then given $m$ i.i.d. draws 
from $\DD$ the error of the firm when 
learning a hypothesis from $\HH$ can be written as
${\err}_{\DD}(\HH) = \err(m, \HH) + \min_{h\in \HH}{\err}_{\DD}(h)$~\cite{sssml}.

The first term, known as \emph{estimation error}, determines how well in 
expectation a model learned with $m$ draws from $\DD$ can 
predict compared to the best model in class $\HH$. 
The second term, known as \emph{approximation error}, determines how well
the best model in class $\HH$ can fit the data generated from $\DD$.

The approximation error is independent of the amount of training data, while 
the estimation error decreases as the volume of training data increases. 
The choice of $\HH$ affects both errors. In particular, fixing the amount of training data, 
increasing the complexity of $\HH$ will increase the estimation error. On the 
other hand, the additional complexity will decrease the approximation error 
as more complicated data generating processes can be fit with more 
complicated models.

Once the amount of data is fixed, the firm can optimize over its 
choice of model complexity to achieve the best error. 
We examine a few widely used ML models and their error forms. 

As a first example, consider the case where the firm is building a neural network and
 has to decide how many nodes $d$ to use. $d$ is the measure 
 of the complexity of the model class and given $m$ data points, the error
 of the model can be written using the following simplification of a result from~\citet{Barron94}.
\begin{lemma}[\citet{Barron94}]
\label{lem:baron}
Let $\HH$ be the class of neural networks with $d$ nodes. Then for any distribution $\DD$, 
with high probability, the error when using $m$ data points to learn a model
from $\HH$ is at most $c_1 d/m + c_2/d,$ for constants $c_1$ and $c_2$.
\end{lemma}

Fixing $m$, the choice of $d$ 
that minimizes the error can be computed by minimizing the bound in Lemma~\ref{lem:baron}
with respect to $d$. This corresponds to $d =  c_2\sqrt{m}/c_1$ and 
we get that the error of the model built by the firm is $\sqrt{c_1c_2/m}$.

As another example, consider the very \emph{simple} setting of \emph{realizable PAC learning}
where the data points are generated by some hypothesis in a fixed hypothesis class. 
\begin{lemma}[\citet{KearnsV94}]
Any algorithm for PAC learning a concept class of VC dimension $d$ must use 
$\Omega(d/\epsilon)$ examples in the worst case.
\end{lemma}
Thus in this setting, in the worst case, firms need $\Theta(1/\epsilon)$ training data 
points to achieve error $\epsilon$. A similar bound gives that with high probability, 
the firms can guarantee error of $\Theta(1/m)$~(see~\cite{KearnsV94}).

In the examples above the error of a firm with $m$ data points
takes the form of either $\Theta(m^{-1/2})$ or $\Theta(m^{-1})$ after the firm optimizes over the choice of 
model complexity. Importantly, the error in both cases (and more generally) degrades  as the number of data points
increases. The rate at which the error degrades is 
commonly known as the \emph{learning rate}. 

There are other learning tasks with learning rates 
different than the examples above. 
Consider a stylized model of a search engine 
where the set of queries is drawn
from a fixed and discrete distribution over a \emph{very large} or even \emph{infinite} set, 
and the search engine can only correctly answer queries that it has seen before. If, as is 
often assumed, the query distribution is heavy-tailed, then the search engine will require a
large training set to return accurate answers. 

In this framework, the probability that a search engine 
incorrectly answers a query drawn 
from the distribution is 
exactly the expectation of the \emph{unobserved} mass of the distribution given the queries 
observed so far. This quantity is known as the \emph{missing mass} of a distribution 
(see e.g.~\cite{BerendK11, Decrouez2018, OrlitskySZ03, Good53}). Lemma~\ref{lem:search}
shows how to bound the expected missing mass for the class of polynomially decaying query distributions.

\begin{lemma}[\citet{Decrouez2018}]
\label{lem:search}
Let $P^k$ for $k>1$ be a discrete distribution with polynomial decay defined over $i\in \N_{\geq 0}$
such that
$\Pr_{x\sim P^k}\left[x = i\right] =  i^{-k}/\Sigma_{j=0}^{\infty} j^{-k}.$
Then the expected missing mass given $m$ draws from $P^k$ is $\Theta(m^{1/k-1})$.
\end{lemma}

By varying  $k$ in the query distribution of
Lemma~\ref{lem:search}, the learning rate in the search problem
can take the form of $\Theta(m^{-i})$ for any $i\in(0,1)$. Thus, the learning rate for search may 
be much faster or slower compared to the previous examples, and the exact rate depends 
on the value of $k$. 

We saw that given a fixed amount of data, a firm using
ML can optimize over its learning decisions to get the best possible error guarantee. 
Furthermore, while error decays as more data becomes available, the rate of decay
can vary widely depending on the task. We next see how various learning rates 
can be incorporated into the parameters of our game. 
\subsection{Error-Based Market Share} \label{sec:errbms}
\label{sec:dynamics}
Consider two competing firms 
(denoted by Firm 1 and 2) that provide identical services 
e.g. search engines. We assume the market shares of the firms depend on 
their ability to make accurate 
predictions e.g. responding to search queries.
As discussed above, the quality of their models is determined ultimately by 
the size of their training data with a task-dependent learning rate. 
Each firm trains a model on its data and uses its model to provide the service. 
Let $\err_1$ and $\err_2$ denote the \emph{excess} error of the firms for the corresponding models.
Intuitively, these errors measure the quality of the firms' services, so a firm 
with smaller error should have higher market share. We assume each firm captures a 
\textit{market share} proportional to the relative errors of the two models. Formally, 
we define Firm 1 and 2's \emph{error-based} market share as

\small{
\begin{equation}
\label{eq:ms}
\ms_1  = 1-\frac{\err_1^a}{\err_1^{\a}+\err_2^{\a}} = \frac{\err_2^{\a}}{\err_1^{\a}+\err_2^{\a}} \text{ and }
\ms_2  = 1-\ms_1.
\end{equation}}\normalsize
	
The constant $\a\in\N$, which we call the \textit{competition exponent} (inspired by
\textit{Tullock contest}~\cite{Tullock01}),
indicates the \emph{ferocity} of the competition, or how strongly a relative difference in the errors of the 
firms' models translates to a market advantage. As $\a$ gets closer to 0, the tendency is 
towards each firm capturing half of the market, and thus a large difference in the models' 
errors is needed for one firm to gain a significant advantage in the market share.  
Conversely, as $\a$ grows larger, even tiny differences in the models' errors translate 
to massive differences in the market share. (See Figure~\ref{fig:comp-exp}.)

\begin{figure}[ht!]
\centering
\includegraphics[width=3in]{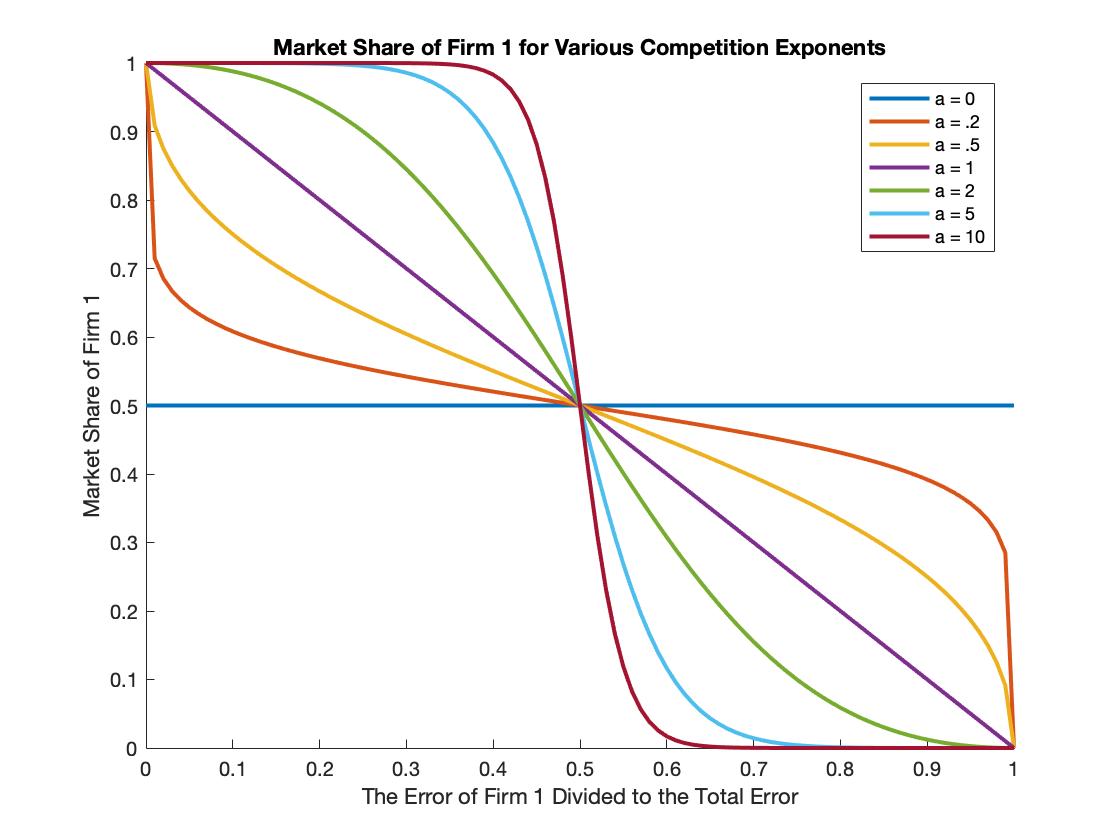}
\caption{\label{fig:comp-exp}Plot of $f=(1-r)^\a/(r^\a+(1-r)^\a)$ for various $\a$ values. 
When $r=\err_1/(\err_1+\err_2)$, $f$ is the market share of Firm 1.}
\end{figure}

An error-based model reflects markets for services which demand extremely low errors, such 
as vision systems for self-driving cars. Under the error-based model, if Firm 1 has $99.99\%$ 
accuracy and Firm 2 has $99\%$ accuracy, Firm 1 will capture $99\%$ of the market share.  
By contrast, an accuracy-based model (i.e. when the market share of Firm 1 is defined as 
$\acc_1^\a/(\acc_1^\a+\acc_2^\a)$) would suggest much less realistic near-even split. 
\begin{figure}[ht!]
\begin{minipage}{0.2\textwidth}
\centering
\begin{tikzpicture}
[scale=0.45, every node/.style={circle,draw=black}, red node/.style = {rectangle, draw=white},]
\node (1) at (0,0){Firm 1};
\node (2) at (4,0){Firm 2};
\draw[->]  [bend right] (1) to (2);
\draw[->] [loop above] (1) to (1);
\draw[->] [bend right] (2) to (1);
\draw[->] [loop above] (2) to (2);
\node [red node] (3) at (2,1.5){$\err_2$};
\node [red node] (4) at (2,-1.5){$\err_1$};
\node [red node] (5) at (0,3.6){$\acc_1$};
\node [red node] (5) at (4,3.6){$\acc_2$};
\end{tikzpicture}
\newline
\end{minipage}
\begin{minipage}{0.25\textwidth}
\centering
\caption{\label{fig:dynamics} Vertices denote the firms and the directed arrows denote 
the probability of transition. $\acc$ is shorthand for accuracy and $\err$ is shorthand for error. }
\end{minipage}
\end{figure}
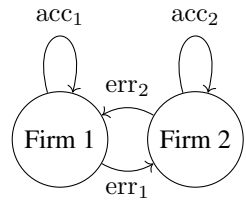

We provide another justification suggesting that an error-based market can arise even when 
the learned model is used to provide an everyday service in which high accuracy is not 
a strict requirement. Consider a customer who, each day, uses the service.  She begins 
by choosing the service of one of the firms uniformly at random. As long as the answers 
she receives are correct, she has no reason to switch to the other firm's service,
and uses the same firm's service tomorrow.  
However, once the firm makes an error,  the customer 
switches to the other firm's service. The transition probabilities are therefore given by 
the accuracy and error of each firm.  
See Figure~\ref{fig:dynamics} for the Markov process representing this example. 

We can think of the market share captured by each firm as the proportion of the days 
on which each firm saw the customer. 
This is exactly the stationary distribution of the associated Markov process as stated
in Lemma~\ref{lem:dynamics}.

\begin{lemma}
\label{lem:dynamics}
Let $\mu_1$ and $\mu_2$ denote the probability mass that the stationary distribution 
of the Markov process in Figure~\ref{fig:dynamics}
assigns to Firms 1 and 2. 
Then $\mu_1 =\err_2/(\err_2+\err_1)$, and $\mu_2=\err_1/(\err_1+\err_2).$
\end{lemma}
\begin{proof}[Sketch of the Proof]
By the definition of a stationary distribution, $\mu_1$ and $\mu_2$ should satisfy the 
following conditions: 
\small{
\begin{align*}
(1-\err_1) \ms_1 + \err_2 \ms_2 & = \ms_1 \\
\err_1 \ms_1 + (1-\err_2) \ms_2 & = \ms_2.
\end{align*}}\normalsize
Given that $\ms_1+\ms_2 = 1$ by definition, we can solve the system of linear
equations to compute the market shares.
\end{proof}

Lemma~\ref{lem:dynamics} states that the market share of each firm in the Markov 
process is exactly the error-based market share as defined in Equation~\ref{eq:ms} 
when setting $\a=1$. A similar argument motivates an error-based market share for 
values of $\a\in\N$, where the customer switches firms after experiencing $\a$ mistakes 
in a row. The probability of making $\a$ mistakes in a row is just $\err_i^\a$ for Firm $i$, 
so the stationary distribution of the Markov process is exactly the two error-based market 
shares as defined in Equation~\ref{eq:ms}.

Using our observations from Section~\ref{sec:ml}, 
we can write the error-based market share in the large-data regime as follows.
\begin{theorem}
\label{lem:reduction}
Let $m_1$ and $m_2$ denote the number of data points of Firm 1 and 2, respectively.
Then for some $\b\in\R^{+}$, the market share of Firm 1 can be written (asymptotically) as 
$\ms_1 = m_1^{\b}/(m_1^{\b}+m_2^{\b})$.
\end{theorem}
\begin{proof}[Sketch of the Proof.]
Depending on the task at hand, we can write the $\err$ of a
firm as $\Theta(m^{-r})$ for some $r\in (0,1]$, where $\err$ refers to the 
excess error of the model with the smallest worst-case error.  Substituting 
this into Equation~\ref{eq:ms} and ignoring lower order terms, which vanish asymptotically,
we get
\small{
\begin{equation*}
\ms_1 = \frac{m_2^{-r\a}}{m_1^{-r\a}+m_2^{-r\a}} = \frac{m_1^{r\a}}{m_1^{r\a}+m_2^{r\a}}.
\end{equation*}}\normalsize
Now let $\b=r\a$. Since $\a$ is a natural number and $r$ is a real
number in $(0,1]$, the combined competition exponent is a real number strictly larger than 0.
\end{proof}

Because $\a$ can be any integer and there exists a corresponding learning problem for any
learning rate in $(0,1]$,
Theorem~\ref{lem:reduction} implies that the combined competition exponent in our game can 
be \emph{any} positive real number,
 motivated by the initial choice
of $\a$ and the learning rate of the firms' ML algorithms. 

The reductions and derivations in Sections~\ref{sec:ml}~and~\ref{sec:dynamics} allow us 
to simplify the acquisition games as follows. We first simplify the actions of each 
firm to only decide whether to buy the data or not, since model choice can be 
optimized once the number of available data points is known. Moreover, 
Theorem~\ref{lem:reduction} not only allows us to simplify the form of market share, but also 
provides us with a meaningful interpretation for any positive (combined) competition exponent. 
\subsection{The Structure of the Game}
\label{sec:game}

Given the reductions so far, we model our 
game as a two-player, one-shot, simultaneous move game.  
Firms 1 and 2 begin the game endowed with an existing number of data points, 
denoted by $\x$ and $\y$, respectively.  Without loss of generality, we assume $\x\geq\y$.  
Each firm must decide whether or not to purchase an 
additional corpus of $\n$ data points\footnote{For simplicity
we assume this data is independent of and identically distributed to the data 
in possession of the firms.}  at a fixed price of $\p$.
The firm can either Buy (denoted by $B$) or Not Buy (denoted by $NB$) the new data.  
If both firms attempt to 
buy the data, the tie is broken uniformly at random (Section~\ref{sec:extensions} discusses 
relaxing the assumption that only one firm may buy the data). 
After the purchase, each firm uses its data to 
train an ML model for its service.  

We assume the particular form
of the market share of Firm 1 using the reduction in Theorem~\ref{lem:reduction}.
The market share of Firm 2 is defined to be one minus the market share of Firm 1.

A \emph{strategy profile} $s$ is a pair of strategies, one for each of the firms. 
Fixing $s$, the utility of Firm $i$ (denoted by $\util_i(s)$) is its market share less
any expenditure. The utility of Firm 1 in all of the strategy profiles of the game is 
summarized in Table~\ref{tab:acq} (rows and columns 
correspond to the actions of Firm 1 and 2). The utility of Firm 
2 is defined symmetrically.  

\begin{table}[h]
\centering
\footnotesize{
\begin{tabular}{ |c|p{.15\textwidth}|p{.15\textwidth}| } 
\hline
Firm 1/Firm 2  & Buy (B) & Not Buy (NB) \\ \hline
Buy (B) & $\tfrac{1}{2}\left(\ms_1(\x+\n, \y, \b)\right.$\newline$+
\left.\ms_1(\x, \y+\n, \b)\right.$\newline$\left. - \p\right)$  &  $\ms_1(\x+\n, \y, \b) - \p$ \\\hline
Not Buy (NB) &  $\ms_1(\x, \y+\n, \b)$ & $\ms_1(\x, \y, \b)$ \\ \hline
\end{tabular}
}
\normalsize
\caption{$\util_1(s)$ in all of the strategy profiles of the game.}
\label{tab:acq} 
\end{table}

A strategy profile is a \emph{pure strategy Nash equilibrium} (pure equilibrium) if no firm can improve its 
utility by taking a different action, fixing the other firm's action. A \emph{mixed strategy 
Nash equilibrium} (mixed equilibrium) is a pair of distributions over the actions (one for each firm) where 
neither firm can improve its expected utility by using
a different distribution over the actions, fixing the other firm's distribution.
We are interested in analyzing the Nash equilibria (equilibria).
\section{Equilibria of the Game}
\label{sec:eq-analysis}

We now turn to finding and analyzing the equilibria. 
First, we introduce some additional notation. Let
\small{\begin{align*}
\A  &=  \frac{\xna}{\xnaya} - \frac{\xa}{\xayna},\\
\C  &=  \frac{\xna}{\xnaya} - \frac{\xa}{\xaya},\\
\D  &=  \frac{\yna}{\xayna} - \frac{\ya}{\xaya}.
\end{align*}}\normalsize
These parameters have intuitive interpretations. $A/2$ is the 
expected change in Firm 1's  
(or Firm 2's) market share when moving the outcome from 
$(NB, B)$ 
(or similarly $(B,NB)$) to $(B, B)$. $\C$ is the change in market share 
that Firm 1 receives if it moves from $(NB,NB)$ to $(B,NB)$, 
and $\D$ is the symmetric relation from the perspective of Firm 2.   

We observe that $\A=\C+\D$. Moreover, since $\C$ and $\D$ are nonnegative,
it is immediately clear that $\A> \max\{\C,\D\}$. 

Finally, when $\x>\y$ (i.e. Firm 1 starts with strictly more data), 
Firm 2 experiences a larger  \textit{absolute}
change in its market share when the outcome changes from 
$(NB,NB)$ to $(NB,B)$ than to 
$(B,NB)$.  In other words, Firm 2 experiences a larger 
\textit{increase} in market share when 
it buys the data compared to the \textit{decrease} it experiences 
when Firm 1 receives the data.
We defer all the omitted proofs to Appendix~\ref{sec:omitted}.
\begin{lemma}
\label{con:c-d}
If $\x > \y$ then for all $\n$ and $\b$ we have that $\C < \D$.
\end{lemma}
\subsection{Characterization of the Equilibria}
\label{sec:nash}

The  equilibria of the game clearly depend on the values of the 
parameters $\x$, $\y$, $\n$, $\p$ and $\b$.  
For example, if $\p>1$ ($p\leq 0$), then neither firm should ever 
want to (not) buy the data. 
We observe that, fixing the values of $\x$, $\y$, $\n$ and $\b$, there 
is a range of values for $\p$ where the data is \emph{too expensive} 
(\emph{too cheap}) and $NB$ ($B$) 
is a dominant strategy for both firms. There is also an intermediate 
range of values for $\p$ where more interesting behaviors emerge, as 
formally characterized in Theorem~\ref{thm:eq}.
\begin{theorem}
\label{thm:eq}
$\newline$
\begin{enumerate}
\item When $\p\leq\max\{\C,\D\}$,  $(B, B)$ is the unique  equilibrium.
\item When $\p \geq \A$,  $(NB, NB)$ is the unique  equilibrium.
\item When $\max\{\C,\D\} < \p < \A$,  $(B, B)$ and $(NB, NB)$ are both  equilibria.
Furthermore, there exists a (unique)  mixed  equilibrium 
$((\alpha,1-\alpha), (\beta,1-\beta))$ such that
\small{
\begin{equation*}
\frac{\alpha}{2(1-\alpha)}  =  \frac{\p-\D}{\A-\p} \text{ and }
\frac{\beta}{2(1-\beta)}  =  \frac{\p-\C}{\A-\p},
\end{equation*}}\normalsize
where $\alpha$ and $\beta$ denote the probabilities that Firms 1 
and 2 select the action $B$, respectively.
\end{enumerate}
\end{theorem}

\begin{proof}
We use flow diagrams to analyze the equilibria of the game 
(see~e.g.~\cite{CandoganMOP11} for more details on this technique).
As a tutorial of this flow diagram argument, we carefully analyze the diagram for 
the regime of our game in which $\p<\min\{\C,\D\}$ as depicted in the top left 
panel of Figure~\ref{fig:eq-regime1}.

In a flow diagram, each vertex corresponds to a strategy profile. 
An arrow indicates that one player changes its strategy while the 
other's action is fixed. 
In particular, in Figure~\ref{fig:eq-regime1} vertical (horizontal) arrows demonstrate 
the change of strategy for Firm 1 (Firm 2).
The numerical value above the arrow indicates how much a player gains by a deviation, 
and arrows are oriented so that they always point in the direction of nonnegative gain. 
The leftmost vertical arrow indicates that Firm 1 increases 
its utility by $(\A-\p)/2$ by changing its decision from $NB$ to $B$, 
fixing that Firm 2 is committed to playing $B$.  Similarly, the rightmost 
vertical arrow indicates that Firm 1 increases its utility by $\C-\p$ when 
it makes this change, fixing that Firm 2 is committed to playing $NB$. 
The horizontal arrows are the symmetric results for Firm 2, fixing the action 
of Firm 1.  The topmost arrow indicates the increase in utility when moving 
from $NB$ to $B$ when Firm 1 plays $B$, and the bottommost 
corresponds to the increase in utility for the same change of action when Firm 1 
plays $NB$.

This particular flow diagram models the regime of the game where the price is 
sufficiently low such that $(B,B)$ is the unique pure  
equilibrium. Consider the profile $(B,B)$. Since arrows only point at, rather than 
originate from, $(B,B)$, unilateral deviations from $(B,B)$ are unprofitable for both players. 
Hence, $(B,B)$ is a pure strategy equilibrium in this regime. Furthermore, 
there is no other pure equilibrium
because there are no other `sinks' in the top left panel of Figure~\ref{fig:eq-regime1}.
Moreover, no mixed equilibrium exists. To see this, note that in a mixed equilibrium, 
a player mixing can only mix over best responses. But since the arrows representing 
Firm 2's deviations both point towards $(B, NB)$ is dominated by $B$; hence $NB$ 
cannot be a best response, so Firm 2 cannot be mixing. But since Firm 1 is not 
indifferent between $B$ and $NB$ if Firm 2 chooses $B$, Firm 1 will not mix either. 
More generally, this logic means that mixed equilibria require arrows pointing in opposite directions.

Similar logic allows us to easily analyze the continuum of games induced as 
allow $p$ varies monotonically. Every value of $p$ induces exactly one of the 
flow diagrams in Figure~\ref{fig:eq-regime1}. Thus, characterizing the equilibria 
in each flow diagram characterizes the equilibria of the different parameter regimes.

\begin{figure*}[ht!]
\centering
\begin{tikzcd}
(B, B) & (B, NB) \arrow{l}[swap]{\frac{1}{2}(\A-\p)}  \\
(NB, B) \arrow{u}{\frac{1}{2}(\A-\p)} & \arrow{l}{\D-\p}  (NB, NB)  \arrow{u}[swap]{\C-\p}
\end{tikzcd}
\begin{tikzcd}
(B, B) & (B, NB) \arrow{l}[swap]{\frac{1}{2}(\A-\p)}  \arrow{d}{\p-\C} \\
(NB, B) \arrow{u}{\frac{1}{2}(\A-\p)} & \arrow{l}{\D-\p}  (NB, NB) 
\end{tikzcd}
\begin{tikzcd}
(B,B) & (B,NB) \arrow{l}[swap]{\frac{1}{2}(\A-\p)}  \arrow{d}{\p-\C} \\
(NB,B) \arrow{u}{\frac{1}{2}(\A-\p)} \arrow{r}[swap]{\p-\D} & (NB,NB) 
\end{tikzcd}
\begin{tikzcd}
(B, B) \arrow{r}{\frac{1}{2}(\p-\A)} \arrow{d}[swap]{\frac{1}{2}(\p-\A)} & 
(B, NB)   \arrow{d}{\p-\C} \\
(NB, B)  \arrow{r}[swap]{\p-\D} & (NB, NB) 
\end{tikzcd}
\caption{\label{fig:eq-regime1} The flow diagrams for different parameter regimes. 
Top left panel when $\p\in (-\infty, \min\{\C,\D\})$.
Top middle when $\p\in (\min\{\C,\D\},\max\{\C,\D\})$.
Top right panel when $\p\in (\max\{\C,\D\}, \A)$. Bottom panel 
when $\p\in (A,+\infty)$.}
\end{figure*}
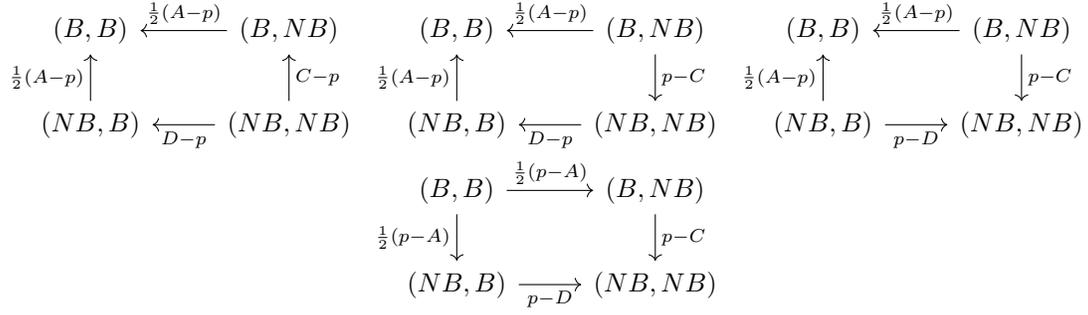
\textbf{(1) $\p\in (-\infty, \min\{\C,\D\})$}: The top left panel of 
Figure~\ref{fig:eq-regime1} represents the flow diagram in this regime
and we can see that the only equilibrium is the pure strategy of $(B, B)$.
	
\textbf{(2) $\p\in(\min\{\C,\D\},\max\{\C,\D\})$}: The top middle 
panel of the Figure~\ref{fig:eq-regime1} represents the flow 
diagram in this regime. By Lemma~\ref{con:c-d},
$(\min\{\C,\D\},\max\{\C,\D\})\equiv (C,D)$.
Again we can see that the only 
equilibrium is the pure strategy $(B, B)$.
	
\textbf{(3) $\p\in(\max\{\C,\D\},\A)$}: The top right panel of 
Figure~\ref{fig:eq-regime1} represents the flow diagram in this regime.
There are two pure equilibria: $(B, B)$ and $(NB, NB)$. 
There also exists a mixed equilibrium.
In a mixed equilibrium, both players are randomizing, 
and thus must be indifferent between the pure 
strategies they are randomizing over; this condition allows us to 
solve for the mixed strategies. 

Let $\alpha$ denote the probability that Firm 1 is playing B. 
Then in a mixed equilibrium, Firm 2 is indifferent between the two actions. Therefore,
\small{
\begin{align*}
&\frac{\alpha}{2}\left(\frac{\yna}{\xayna} + \frac{\ya}{\xnaya}-\p\right)\\
&\qquad+\left(1-\alpha\right) \left(\frac{\yna}{\xayna}-\p\right) = \\
&\alpha\left(\frac{\ya}{\xnaya}\right)  + \left(1-\alpha\right)\left(\frac{\ya}{\xa+\ya}\right).
\end{align*}}\normalsize
By rearranging we get that
\small{
\begin{align*}
\frac{\alpha}{2(1-\alpha)}
= \frac{\p-\D}{\A-\p},
\end{align*}}\normalsize
as claimed. 
	
Similarly let $\beta$ denote the probability that Firm 2 is playing B. 
Then in a mixed equilibrium, Firm 1 is indifferent between 
the two actions. With a similar calculation we can show that
\small{
\begin{align*}
\frac{\beta}{2(1-\beta)} 
&= \frac{\p-\C}{\A-\p},
\end{align*}}\normalsize
as claimed.
	
\textbf{(4) $\p\in(A,+\infty)$}: The bottom panel of Figure~\ref{fig:eq-regime1} 
represents the flow diagram in this regime
and we can see that the only equilibrium is the pure strategy of $(NB, NB)$.
\end{proof}
Theorem \ref{thm:eq} allows us to make several key observations about the market 
structure of this game. First, since $C$ and $D$ represent the maximum increase in 
the market share firms could achieve by buying the data, the fact that the only equilibrium 
when $p \in [\min\{C,D\},\max\{C,D\}]$ is $(B,B)$ means that both firms buy the data despite 
the fact that the \emph{best-case improvement} in market share is less than what they pay. 
This `race for data' thus has the character of a prisoner's dilemma -- if both firms could agree 
not to buy the data, they would be better off, but either would be tempted to buy the data and 
improve market share. 

Second, Theorem \ref{thm:eq} illustrates how several features of equilibrium depend on the 
ferocity of competition, as determined by the exponent $\a$; 
as $\a$ varies, the frontiers of the regimes described in Theorem \ref{thm:eq} shift too. 
For example, in the case that $\a=0$, market 
share is split evenly between the two firms, regardless of error or accuracy; unsurprisingly, 
as $\a \rightarrow 0$ (which implies $\b \rightarrow 0$), $\A, \C$, and $\D$ also approach 
$0$, so the payoff difference between strategy profiles becomes negligible. 
As a consequence, the regimes (1), (2), and (3) collapse, and all but very small $p$ induce 
regime (4), where $(NB,NB)$ is the only equilibrium. Thus, for small $\a$, 
unless $\p$ is very close to zero, $(NB,NB)$ is the only equilibrium. We observe similar behavior 
when $\a$ is large. Assuming that $\x>\y+\n$, then $\a \rightarrow \infty$ implies that 
$\A \rightarrow 0$ (and hence $\b \rightarrow \infty$), 
again implying that regimes (1), (2), and (3) collapse. Thus again, unless $\p$ is very 
close to $0$, $(NB, NB)$ 
is the unique equilibrium. This is for a different reason than the $\a$ small case, however: Firm 
1 now has no incentive to buy, since it is guaranteed almost the whole market share using its 
current model. Moreover, in this scenario, Firm 2's initial disadvantage is too great to be 
overcome by buying the data. 

If $\a$ is in between these two extremes, many choices of $\x$ and $\y$ lead to a 
non-empty interval $(\max\{\C, \D\}, \A)$, with endpoints far from $0$ and $1$. When $\p$ 
falls in this interval, regime (2) holds, so a mixed 
equilibrium exists; we discuss solving for  this mixed equilibrium in Section~\ref{sec:eqanalysis-mono}.
The complete characterization of the equilibria for all regimes of $p$ in Theorem \ref{thm:eq} 
allows us to pin down the optimal fixed price from the perspective of maximizing seller's revenue. 
However, in full generality, the seller's problem encompasses further possibilities like auction pricing; 
hence, we defer this calculation to future work. See Section~\ref{sec:future} for a discussion. 
\subsection{Mixed Equilibrium and Monotonicity Analysis}
\label{sec:eqanalysis-mono}

Next, we carefully examine the mixed equilibrium and study the relationship between 
the weights each firm places on each action and the parameters of the game.

Recall that $\alpha$ and $\beta$ in Theorem~\ref{thm:eq} denote the probability 
that Firms 1 and 2 purchase the data in the mixed equilibrium. 
When $\x>\y$, then $\alpha < \beta$ which implies that the 
smaller firm will succeed more often in purchasing the data in the mixed 
 equilibrium. The relationship of $\alpha$ and $\beta$ with the number of data points 
$\n$ and the price $\p$ is as follows.
\begin{lemma}
\label{lem:mixed-eq-mono}
Let $((\alpha,1-\alpha), (\beta,1-\beta))$ denote the mixed equilibrium in 
the the regime where $\max\{\C,\D\} < \p < \A$.
Then $\alpha$ and $\beta$, both increase when $\p$ increases or $\n$ decreases.
\end{lemma}

Lemma \ref{lem:mixed-eq-mono} may seem counterintuitive as
it implies that as the price $\p$ \textit{rises} through the range in which a mixed equilibrium 
exists, the probability that any of the firms want to buy the data \textit{also increases}.  
However, once the price $p$ crosses the threshold $\A$, the unique  equilibrium is the pure 
strategy $(NB, NB)$. This gives rise to a discontinuity. See Figure~\ref{fig:dis}.

\begin{figure}[ht!]
\centering
\includegraphics[width=3in]{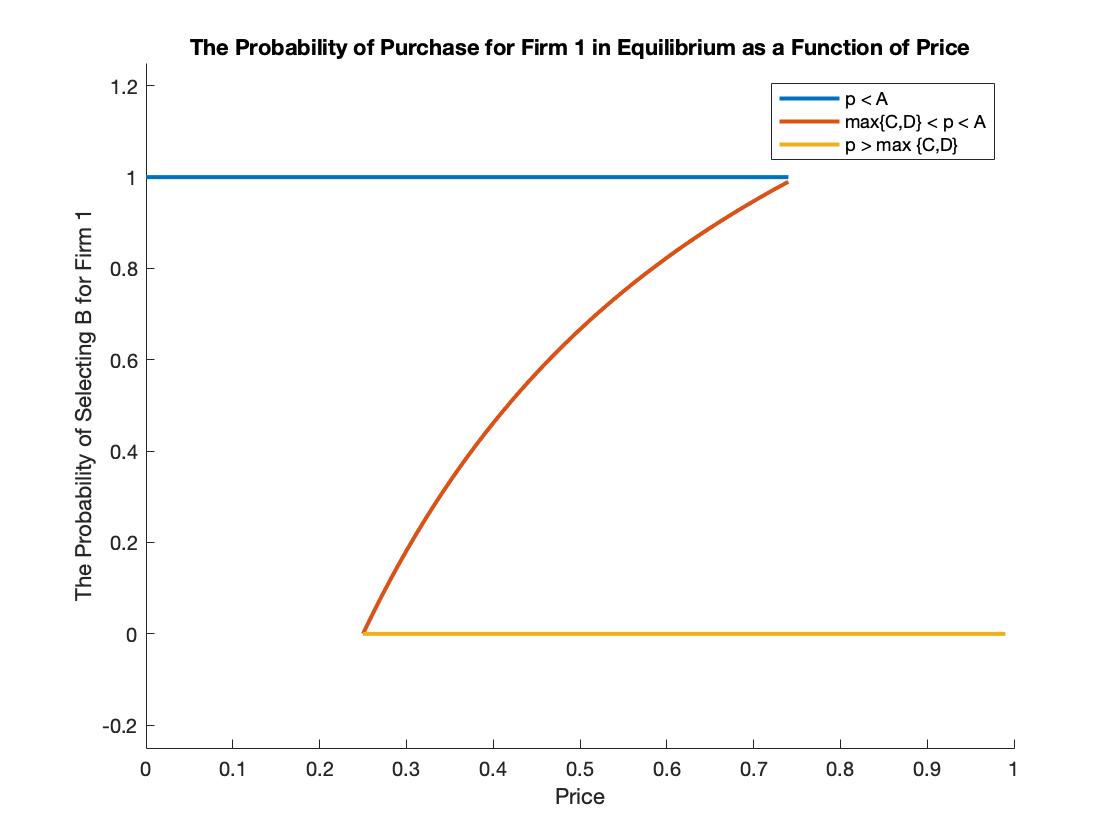}
\caption{\label{fig:dis} The probability of selecting $B$ in equilibrium for 
Firm 1 as a function of $p$. 
The blue, red and yellow lines correspond to the $(B,B)$, mixed and $(NB, NB)$ 
equilibria.}
\end{figure}

Of course, this all says nothing about the equilibrium utilities for the firms; 
as long as the equilibrium utilities are not identical, players will naturally have ordinal 
preferences over the set of equilibria. We analyze these preferences in 
Lemma~\ref{lem:tradewar}, which elucidates the discontinuity at $\p=\A$.

\begin{lemma}
 \label{lem:tradewar}
When $\p\in(\max\{\C,\D\},\A)$, $\util_1(NB, NB) \geq \util_1(s)$  
and $\util_2(NB, NB) \geq \util_2(s)$
for all strategy profiles $s$.  
However, $\util_1(B, NB) \geq \util_1(B, B) \geq \util_1(NB, B)$, 
while $\util_2(NB, B) \geq \util_2(B, B) \geq \util_2(B,NB)$. 
\end{lemma}

While both firms agree that $(NB,NB)$ is the most preferred outcome, their preferences 
over the remaining three outcomes are discordant.  In particular, given that at least one 
firm will try to buy the data, each firm would prefer itself to be the buyer rather than the 
opponent.  If either firm believes the other may try to buy the data, it will put positive weight 
on the action $B$ in the mixed equilibrium.  Once the price crosses $\A$, both firms know 
that it would be irrational for the other to buy, so we see a unique pure equilibrium of $(NB,NB)$. 
\subsection{Change in the Market Share}
\label{sec:ms-change}

We now analyze the change in the market shares.
\begin{lemma}
\label{lem:ms-change}
When $\x\geq\y$, 
the only strategy profile that strictly increases 
the market share of Firm 1 
is $(B, NB)$. 
\end{lemma}

So while only $(B, NB)$ leads to an increase in the market share
of Firm 1, it is not a pure equilibrium. 
We show that  
even when firms play according to the mixed equilibrium, 
the expected market share of Firm 1 does not strictly increase.
\begin{theorem}
\label{thm:ms-change}
When $\x\geq \y$ and $\p\in(\max\{\C,\D\}, A)$,
the expected market share of Firm $1$ does not strictly increase if 
both firms play according to the mixed equilibrium.
\end{theorem}

Together, Lemma~\ref{lem:ms-change} and Theorem~\ref{thm:ms-change} 
demonstrate that the natural forces 
of the interaction on the market are, perhaps surprisingly, antimonopolistic. 
Since we assume that Firm 1 enters 
the game with a greater market share than Firm 2, but that no equilibrium 
allows Firm 1 to increase its market 
share, the game disfavors the concentration of market power.  This raises the 
question of whether this antimonopolistic tendency is good for the users.  
We analyze this next.
\subsection{Consumer Welfare in Equilibrium}
\label{sec:eq-analysis-cons}

We now consider the perspective of the users of the firm's service.  
We show that consumers prefer the outcome $(B,NB)$, in which 
the initially stronger firm concentrates its market power.  This is not 
supported by a pure equilibrium in any regime, nor is it the most likely outcome 
generated by mixed equilibrium; hence, we will see that the interests of the 
firms do not align with the interests of the 
consumers. We define the consumer welfare as follows.
\begin{definition}
\label{def:cons-welfare}
Let  $m_1(s)$ and $m_2(s)$ denote the
(expected) number of data points that Firm 1 and 2 posses when 
playing according to strategy profile $s$.
Then the \emph{consumer welfare} is
\begin{equation*}
\footnotesize
\cons(s) = \ms_1(s) (1-\err_1(m_1(s)))+ \ms_2(s)(1-\err_2(m_2(s))).
\end{equation*}\normalsize
\end{definition}
The welfare definition arises from assuming consumers receive $1$ unit of 
utility for accurate predictions and $0$ for erroneous ones. Notice that maximizing 
this definition of consumer welfare is exactly equivalent to minimizing the 
market-share weighted error probability. This leads to Theorem~\ref{thm:consumer-pref}.
\begin{theorem}
\label{thm:consumer-pref}
Suppose $\x > \y$. 
Then the consumers have the following
preferences over the strategy profiles.
\footnotesize{
\begin{equation*}
\cons(B, NB) >  \cons(B, B) > \cons(NB, B) > \cons(NB, NB).
\end{equation*}}\normalsize
\end{theorem}
Note that consumers' preference for the outcome in which Firm 1 concentrates 
its market power is \emph{not} the same as saying that the consumers prefer a monopoly. 
Rather, the consumers have a preference for higher quality services.  
When $\x>\y$, Firm 1's model before acquiring the data has a lower 
error rate than that of Firm 2, and so, of all the possible outcomes, the one 
which leads to a product with the lowest error rate is the one in which Firm 1 
is able to improve on its already superior product.  But if Firm 2 were 
not a player at all, then a monopolistic Firm 1 would \emph{have no incentive to} 
buy the new data. Therefore, a monopoly without the threat of competition
will not lead to the best outcome from the consumer's perspectives. 
\section{Extensions and Robustness}
\label{sec:extensions}
Next, we consider robustness to three simple extensions. 
\paragraph{Firm Acquisition} We treat the data seller as a market participant 
with its own customers and market share. This allows us to model firm acquisition: 
buying the data translates to acquiring the firm and its customers, and neither firm 
buying the data corresponds to the third firm remaining in the market. 
\paragraph{Simultaneous Sale} Rather than the data being exclusively sold to one 
firm in the case that both firms buy, we allow the seller to sell the data to both firms 
at the same fixed price. 
\paragraph{Correlated Equilibria}  We consider the richer concept of \emph{correlated 
equilibrium} and search for additional equilibria. 

In each of these extensions, we can again derive the quantities $\A$, $\C$, and $\D$; 
while the precise quantities change, their rankings and relationships do not. Thus, in 
the first two extensions, the general phenomenon of three regimes, with mixing over the 
middle regime, remains unchanged. In the third extension, some new correlated equilibria 
exist, but none include the qualitatively different result of coordinating purchase of the data 
by a single firm. Moreover, in expectation, the market share becomes less asymmetric in all 
extensions.

\section{Future Directions}
\label{sec:future}
We view our work as a first step towards modeling and analyzing competition for data in 
markets driven by ML.
There are several directions for further investigation. First, we modeled the data to be acquired 
as having a fixed size and a fixed price, but real datasets can be divisible.  One further 
direction to consider is a game in which we expand the strategy space of the players to 
include buying any number of data points at a fixed price \textit{per data point} or nonlinear 
function of the number of data points purchased.  More generally, treating the seller of the 
data as an additional player in the game allows for further questions, such as: what is the 
optimal revenue-generating mechanism to sell the data? And does the optimal mechanism 
maximize social welfare? 

Additionally, many firms that provide learning-based services acquire their data through their 
customers that use the service.  In this way, capturing a larger market share induces a feedback 
loop which allows a firm to iteratively improve its product.  What can be said about our game in a 
repeated setting with dynamic feedback effects?  Furthermore, firms that provide digital services 
often operate in a secondary market in which other firms pay for advertising spots in their product. 
Improving one's market share should in principle allow a firm to charge advertisers a higher price, 
but we do not know to what extent this affects the analysis of the equilibria of the game. Incorporating 
advertiser behavior would greatly complicate the model but provide potentially interesting results.
\bibliographystyle{named}
\bibliography{bib}
\appendix
\section{Omitted Proofs}
\label{sec:omitted}

\noindent\textit{Proof of Lemma~\ref{con:c-d}.}
Define $u =\x/\n$ and $v = \y/\n$. Then $u>v$ and we also have that 
\small{
\begin{align*}
\C &= \frac{(1+u)^\b}{(1+u)^\b+v^\b}-\frac{u^\b}{u^\b+v^\b}, \\
\D &= \frac{(1+v)^\b}{u^\b+(1+v)^\b}-\frac{v^\b}{u^\b+v^\b}.
\end{align*}}\normalsize
Next let $U = u^\b$, $V=v^\b$, $W = (1+u)^\b-u^\b$ and 
$Z=(1+v)^\b-v^\b$. Notice that $W>0$ and $Z>0$ for all $b>0$.
Algebraic manipulations show that
\small{
\begin{align*}
&\C < \D \\
&\Leftrightarrow \frac{(1+u)^\b}{(1+u)^\b+v^\b}-\frac{u^\b}{u^\b+v^\b} <
 \frac{(1+v)^\b}{u^\b+(1+v)^\b}-\frac{v^\b}{u^\b+v^\b}\\
&\Leftrightarrow \frac{U+W}{U+W+V} - \frac{V+Z}{V+Z+U} < \frac{U-V}{U+V} \\
&\Leftrightarrow  VW(U+V+Z) < UZ(U+V+W).
\end{align*} }\normalsize

Fix a pair of $(u,v)$ with $u>v$, there are two cases to consider\footnote{These 
two cases correspond to $b\geqslant1$ and $0<b<1$, but this 
correspondence is irrelevant.}: 
(1) $W \geq Z$ and (2) $W<Z$.

In case (1)
$U+V+Z \leq U+V+W$. Notice that $VW <UZ$ would suffice to prove the 
claim in this case, because
\small{
\begin{align*}
VW < UZ &\implies VW(U+V+W) < UZ(U+V+W) \\&
\implies VW(U+V+Z) < UZ(U+V+W).
\end{align*}}\normalsize
which is the last condition in the chain of double implications. In fact, 
$VW < UZ$ does hold, because   
\small{
\begin{align*}
u>v &\implies v(1+u) < u(1+v) \\
&\implies v^\b(1+u)^\b < u^\b(1+v)^\b \\
&\implies v^\b((1+u)^\b-u^\b) < u^\b((1+v)^\b-v^b)\\
&\implies VW < UZ.
\end{align*}}\normalsize

Now we turn to the second case. Suppose $W < Z$. We again must show that 
$VW(U+V+Z) < UZ(U+V+W)$. Notice that the following implication holds.
\small{\begin{align*}
VW(U+V+Z) &< UW(U+V+W) \implies\\ 
VW(U+V+Z) &< UZ(U+V+W),
\end{align*}}\normalsize
since \small{$$UW(U+V+W) < UZ(U+V+W).$$}\normalsize
Hence we show that the first inequality is true. 
Note that
\small{
\begin{align*}
VW(U+V+Z) &< UW(U+V+W)\\
&\iff V(U+V+Z) < U(U+V+W) \\
&\iff V(V+Z) < U(U+W) \\
&\iff v^{\b} (1+v)^{\b} < u^{\b} (1+u)^{\b},
\end{align*}}\normalsize
which is trivially true by $u >v$ and $b>0$. 

Hence in both cases, we have that \small{$$VW(U+V+Z) < UZ(U+V+W) 
\implies C<D,$$\normalsize 
concluding the proof. Finally, we note that symmetry yields the corresponding 
claim $\y>\x \implies D <C$. 
\qed

\vspace*{0.5em}

\noindent\textit{Proof of Lemma~\ref{lem:mixed-eq-mono}.}
The left hand sides of the equations characterizing the mixed equilibrium in the 
statement of Theorem~\ref{thm:eq} 
have the form  $c/(2(1-c))$ which is increasing in $c$ when $c\in(0,1)$. 
If we call the left hand side of either of these equations $\ell$, we can solve 
for $c= 2\ell/(1+2\ell)$
which is also monotonically increasing in $\ell$ over $\ell \in [0,1]$. Hence, 
to analyze the 
monotonicity of the $c$, it is suffices to analyze the monotonicity of $\ell$. 
But since $\ell$ must also equal the right hand side, 
it suffices to analyze the monotonicity of the right hand sides of these equations. 
	
Since $\p$ does not appear in $\b,\C$ or $\D$, it is easy to see that both fractions 
$(\p-\D)/(\A-\p)$ and $(\p-\C)/(\A-\p)$ increase as $\p$ increases. Hence, both 
$\alpha$ and $\beta$ increase as $\p$ increases in the regime $\p\in(\max\{\C,\D\},\A)$. 
The parameter $\n$ on the other hand appears in $\A$, $\C$ and $\D$. All of 
these parameters increase as $\n$ increases. It is then similarly easy
to see that both fractions $(\p-\D)/(\A-\p)$ and $(\p-\C)/(\A-\p)$ decrease as  
$\n$ increases. 
\qed

\vspace*{0.5em}

\noindent\textit{Proof of Lemma~\ref{lem:tradewar}.}
We claim that the ordinal preferences of Firm 1 over the outcomes
are as follows.
\small{
\begin{equation*}
\util_1(NB, NB) \geq \util_1(B, NB) \geq \util_1(B, B) \geq \util_1(NB, B).
\end{equation*}}\normalsize
Since \small{$\util_1(B, B) = (\util_1(B, NB)+\util_1(NB, B))/2$}\normalsize, it suffices 
to show that \small{$$\util_1(NB, NB) \geq \util_1(B, NB) \geq \util_1(NB, B).$$}\normalsize 

We first show that 
$\util_1(NB, NB) \geq \util_1(B, NB)$.
\small{\begin{align*}
\p-\C\geq 0 &\implies  \p-\frac{\xna}{\xnaya}+\frac{\xa}{\xa+\ya}\geq 0\\  
& \implies \frac{\xa}{\xa+\ya}\geq \frac{\xna}{\xnaya}-\p \\
&\implies \util_1(NB, NB) \geq \util_1(B, NB).
\end{align*}
We then show that $\util_1(B, NB) \geq \util_1(NB, B)$.
\begin{align*}
\A-\p\geq 0
&\implies\frac{\xna}{\xnaya}-\frac{\xa}{\xayna}-\p \geq 0 \\
&\implies \frac{\xna}{\xnaya}-\p\geq \frac{\xa}{\xayna} \\
&\implies \util_1(B, NB)\geq \util_1(NB, B).
\end{align*}}\normalsize

Moreover,  the ordinal preferences of Firm 2 are as follows.
\small{\begin{equation*}
\util_2(NB, NB) \geq \util_2(NB, B) \geq \util_2(B, B) \geq \util_2(B,NB).
\end{equation*}}\normalsize
Again note that $\util_2(B, B) = (\util_2(NB, B)+\util_2(B, NB))/2$. So it suffices 
to show that $\util_2(NB, NB) \geq \util_2(NB, B) \geq \util_2(NB, B)$. 

We first show that $\util_2(NB, NB) \geq \util_2(NB, B)$.
\small{\begin{align*}
\p-\D\geq 0 &\implies  \p-\frac{\yna}{\xayna}+\frac{\ya}{\xa+\ya}\geq 0\\  
& \implies \frac{\ya}{\xa+\ya}\geq \frac{\yna}{\xayna}-\p \\
&\implies \util_2(NB, NB) \geq \util_2(NB, B).
\end{align*}}\normalsize

We wrap up by showing that $\util_2(NB, B) \geq \util_1(B, NB)$.
\small{\begin{align*}
\A-\p\geq 0
&\implies  \frac{\xna}{\xnaya}-\frac{\xa}{\xayna}-\p \geq 0 \\
&\implies \left(1-\frac{\ya}{\xnaya}\right)-\\
&\qquad\qquad\qquad\left(1-\frac{\yna}{\xayna}\right)-\p \geq 0 \\
&\implies \frac{\yna}{\xayna}-\p \geq \frac{\ya}{\xnaya} \\
&\implies \util_2(NB, B) \geq \util_1(B, NB).
\end{align*}}\normalsize
\qed

\vspace*{0.5em}

\noindent\textit{Proof of Theorem~\ref{thm:ms-change}.}
Let $\alpha$ and $\beta$ be as  in Theorem~\ref{thm:eq}. 
First observe that by rewriting the conditions for the mixed equilibrium,
we get that
\small{\[
\alpha = \frac{2\left(\p-\D\right)}{\A+\p-2\D} \text{ and } 
\beta = \frac{2\left(\p-\C\right)}{\A+\p-2\C}.
\]}\normalsize
Then in the mixed equilibrium, the four outcomes occur with the 
following probabilities:
\begin{enumerate} 
\item $(B, NB)$ with probability $\alpha(1-\beta)$, 
\item $(NB, B)$ with probability $(1-\alpha)\beta$, 
\item $(NB, NB)$ with probability $(1-\alpha)(1-\beta)$, and 
\item $(B,B)$ with probability $\alpha\beta$.
\end{enumerate}
The expected change in Firm 1's market share can be
calculated by summing over the change in its market share in each outcome
(see the proof of Lemma~\ref{lem:ms-change}), weighted 
by how often the outcome occurs
in the mixed equilibrium. Thus the expected change in Firm 1's 
market share is
\small{\begin{align*}
\alpha\left(1-\beta\right) &\C + \left(1-\alpha\right)\beta \D + 
\left(1-\alpha\right)\left(1-\beta\right) 0 + \alpha\beta \frac{\C-\D}{2}\\
& = 2\left(\frac{\left(\C-\D\right)\left(\p\left(\A-\C-\D\right)+\C\D\right)}
{\left(\A+\p-2\C\right)\left(\A+\p-2\D\right)}\right).
\end{align*}}\normalsize
Since we are only interested in whether the market share increases or decreases,
we only care about the sign of the above term. 
The denominator is
always positive, as both terms in the denominator are positive 
when $\p\in(\max\{\C, \D\}, \A)$. So it suffices to show that the numerator 
is non-negative.
	
When $\x=\y$ then the first term in the numerator is zero, so the
expected market share is the same as the initial market share. 
On the other hand
when $\x>\y$, Lemma~\ref{con:c-d} implies that the first term in the 
numerator is negative. We claim that 
the second term in the numerator is always positive. To see this, 
first observe that $\p\left(\A-\C-\D\right)+\C\D$ is
a linear function of $\p$ and it is strictly positive at both end 
points $\p=\max\{\C,\D\} = \C$ 
(by Lemma~\ref{con:c-d}) and 
$\p=\A$. By the properties of linear functions, 
the term is positive for all 
values of $\p$ between $\max\{\C,\D\}$ and $\A$, which is 
exactly the regime we are interested in.
\qed

\vspace*{0.5em}

\noindent\textit{Proof of Lemma~\ref{lem:ms-change}.}
The change in the market share of Firm 1 in strategy profiles $(B, NB)$ 
compared to the beginning of the game is the parameter $\C$, which is 
always positive. 
Similarly, the change in the market share of Firm 2 in strategy profile $(NB, B)$ 
compared to the beginning of the game is the parameter $\D$. 
The change in the market share of Firm 1 in this strategy profile is $-\D$ since 
the sum of market shares is always one, and since $\D$ is always positive, 
$-\D$ is negative. Moreover, the expected change in Firm 1's market share 
for $(B, B)$ is $(\C-\D)/2$ because we decide which firm purchases the 
data by a fair coin toss.  By Lemma~\ref{con:c-d}, $\D\geq \C$, 
so $(\C-\D)/2\leq 0$.   Finally, there is no change in the market
shares in the outcome $(NB, NB)$.
Thus, only for $(B,NB)$ does Firm 1's market share strictly increase.
\qed

\vspace*{0.5em}

\noindent\textit{Proof of Theorem~\ref{thm:consumer-pref}.}
We first simplify the consumer welfare for a strategy profile $s$ 
by shorthands $\err_1 \equiv \err_1(m_1(s))$, 
$\err_2 \equiv \err_1(m_2(s))$, $\ms_1\equiv\ms_1(s)$ and $\ms_2\equiv\ms_2(s)$.
\small{
\begin{align*}
\cons(s) &= \ms_1 \left(1-\err_1\right)+ \ms_2\left(1-\err_2\right)\\
& = \frac{\err_2^\b}{\err_1^\b+\err_2^\b} \left(1-\err_1\right)+ 
\frac{\err_1^\b}{\err_1^\b+\err_2^\b}\left(1-\err_2\right) \\
&= 1-\frac{\err_1^\b\err_2+\err_2^\b\err_1}{\err_1^\b+\err_2^\b}.
\end{align*}}\normalsize
So the strategy profile that maximizes the social welfare of the consumers 
equivalently maximizing
the following equation
\small{
\begin{align*}
\max_s\enspace\cons(s) \equiv \max_{s} \frac{\err_1^\b+\err_2^\b}
{\err_1^\b\err_2+\err_2^\b\err_1}.
\end{align*}}\normalsize

We take the following three steps to prove the statement of the theorem: 
(1) $\cons(B, NB) > \cons(NB, B)$, (2) $\cons(B, B) = \left(\cons(NB, B)
+\cons(B, NB)\right)/2$ 
and (3) $\cons(NB, NB) < \cons(s)$ for all $s\neq (NB, NB)$. For simplicity
in the rest of the proof we assume that the error scales with the square root 
of the number of
data points.

To prove part (1), first, observe that
\small{
\begin{align*}
&(\x+\n)^{(\b-1)/2} \x^{\b/2}\left(\sqrt{\x+\n}-\sqrt{\x}\right)\\
-&(\y+\n)^{(\b-1)/2} \y^{\b/2}\left(\sqrt{\y+\n}-\sqrt{\y}\right) > 0,
\end{align*}}\normalsize
since 
\small{
$$(\x+\n)^{(\b-1)/2} \x^{\b/2}>(\y+\n)^{(\b-1)/2} \y^{\b/2},$$}\normalsize
when $\x>\y$ and also
\small{
$$\left(\sqrt{\x+\n}-\sqrt{\x}\right) > \y^{\b/2}\left(\sqrt{\y+\n}-\sqrt{\y}\right),$$}\normalsize
since the function $f(w) = \sqrt{w+b} - \sqrt{w}$ is increasing in $w$ when $b>0$.
Adding a positive term to the expression above, we get that
\footnotesize{
\begin{align*}
&(\x+\n)^{(\b-1)/2} \x^{\b/2}\left(\sqrt{\x+\n}
-\sqrt{\x}\right)\\
&-(\y+\n)^{(\b-1)/2} \y^{\b/2}\left(\sqrt{\y+\n}-\sqrt{\y}\right) \\
&+ \left((\x+\n)^{(\b-1)/2} (\y+\n)^{(\b-1)/2}
-\x^{(\b-1)/2} \y^{(\b-1)/2}\right)\times\\
&\left(\sqrt{\x}-\sqrt{\y}\right) > 0\\
\implies & (\x+\n)^{\b/2} \x^{(\b-1)/2}+(\x+\n)^{\b/2}(\y+\n)^{(\b-1)/2}\\
&+\x^{(\b-1)/2}\y^{\b/2}+(\y+\n)^{(\b-1)/2}y^{\b/2} \\
&> (\x+\n)^{(\b-1)/2}\x^{\b/2}+(\x+\n)^{(\b-1)/2}(\y+\n)^{\b/2}\\
&+\x^{\b/2}\y^{(\b-1)/2}+(\y+\n)^{\b/2}\y^{(\b-1)/2} \\
\implies & \frac{(\x+\n)^{\b/2}+\y^{\b/2}}{(\x+\n)^{(\b-1)/2}+\y^{(\b-1)/2}} 
> \frac{\x^{\b/2}+(\y+\n)^{\b/2}}{\x^{(\b-1)/2}+(\y+\n)^{(\b-1)/2}} \\
\implies & \frac{(\x+\n)^{-\b/2}+\y^{-\b/2}}{(\x+\n)^{-\b/2}\sqrt{\y}
+\y^{-\b/2}\sqrt{(\x+\n)}} \\ &> \frac{\x^{-\b/2}+(\y+\n)^{-\b/2}}{\x^{-\b/2}\sqrt{\y+\n}
+(\y+\n)^{-\b/2}\sqrt{\x}} \\
\implies & \cons(B, NB) > (NB, B).
\end{align*}}\normalsize
In the penultimate step we multiplied the numerator and denominator 
of the left and right hand fractions
by $(\x+\n)^{-\b/2}\y^{-\b/2}$ and $(\y+\n)^{-\b/2}\x^{-\b/2}$, respectively.

Part (2) is trivial given that in our model when both players choose $B$ one of 
the strategy profiles 
$(B, NB)$ and $(NB, B)$ is chosen with equal probability. Since 
$\cons(B, NB) > \cons(NB, B)$ 
by part (1) then part (2) implies that $\cons(B, NB) >  \cons(B, B) > \cons(NB, B)$.
To prove part (3), we show
that $\cons$ is increasing in $\y$. 
This implies that $\cons(NB, B) > \cons(NB, NB)$.
To show that the function $\cons$ is increasing in $\y$ we show that 
the partial derivative with respect to $\y$ is positive.
\small{\begin{align*}
&\frac{\partial \cons(NB, NB)}{\partial \y}  =  \frac{\partial}{\partial \y} 
\left(\frac{\x^{-\b/2}+\y^{\b-/2}}{\x^{-\b/2}\sqrt{\y}+\y^{-\b/2}\sqrt{\x}}\right)\\
 = &\left(\x^{-\b/2}\sqrt{\y}+\y^{-\b/2}\sqrt{\x}\right)^{-2}\times\\&
\Big(-\frac{\b}{2}\x^{-(\b+1)/2}\left((\x^{-\b/2}\sqrt{\y}+\y^{-\b/2}\sqrt{\x}\right)\\
&\hspace{2mm}+\frac{\b}{2}\left((\x^{-\b/2}+\y^{-\b/2}\right)+\frac{x^{-3/2}}{2}\left(\x^{-\b/2}
+\y^{-\b/2}\right)\Big)\\
 = &\frac{1}{2\x\left(\x^{-\b/2}\sqrt{\y}+\y^{-\b/2}\sqrt{\x}\right)^2}\times\\
&\Big(\b\x^\b\left(1-\frac{1}{\sqrt{\y}}\right)+
\b\x^{-\b/2}y^{-\b/2}\left(1-\frac{1}{\sqrt{\x}}\right)\\&+\frac{1}{\sqrt{\x}}
\left(\x^{-\b/2}+\y^{-\b/2}\right)
\Big) > 0,\\
\end{align*}}\normalsize
which is positive since both $\x \geq 1$ and $\y \geq 1$.
\qed
\end{document}